\documentclass[10pt,journal,onecolumn]{IEEEtran}
\usepackage{cite}

\usepackage{hyperref}
\usepackage{amsfonts}
\usepackage{amssymb}
\usepackage{mathrsfs}
\usepackage{amsmath,bbold}
\usepackage{amsthm}
\usepackage{color}

\usepackage{subfigure} 
\usepackage{xcolor}
\usepackage{array,float}

\usepackage{optidef} 

\newtheorem{theorem}{Theorem} 
\newtheorem{lemma}[theorem]{Lemma}

\newtheorem{definition}[theorem]{Definition}
\newtheorem{remark}{Remark}

\newcommand{\ket}[1]{\ensuremath{\vert#1\rangle}}

\newcommand{\tr}{\ensuremath{\mathrm{tr}}}

\newcommand{\Real}{\ensuremath{\mathrm{Re}}}
\newcommand{\Imag}{\ensuremath{\mathrm{Im}}}

\def\>{\rangle} 
\def\<{\langle}

\begin{document}

\title{Trade-offs on number and phase shift resilience in bosonic quantum codes}





\author{ 
    \IEEEauthorblockN{Yingkai Ouyang\IEEEauthorrefmark{1}\IEEEauthorrefmark{2}, Earl T. Campbell\IEEEauthorrefmark{1}\IEEEauthorrefmark{3}}\\
    \IEEEauthorblockA{\IEEEauthorrefmark{1}Department of Physics \& Astronomy, University of Sheffield, Sheffield, S3 7RH, United Kingdom
    }\\
    \IEEEauthorblockA{\IEEEauthorrefmark{2}Department of Electrical and Computer Engineering, National University of Singapore, Singapore
    }\\
    \IEEEauthorblockA{\IEEEauthorrefmark{3}AWS Center for Quantum Computing, Pasadena, CA 91125 USA
    }
}

\date{\today}

\maketitle

\begin{abstract}
Quantum codes typically rely on large numbers of degrees of freedom to achieve low error rates. However each additional degree of freedom introduces a new set of error mechanisms. Hence minimizing the degrees of freedom that a quantum code utilizes is helpful.
One quantum error correction solution is to encode quantum information into one or more bosonic modes. We revisit rotation-invariant bosonic codes, which are supported on Fock states that are gapped by an integer $g$ apart, and the gap $g$ imparts number shift resilience to these codes. Intuitively, since phase operators and number shift operators do not commute, one expects a trade-off between resilience to number-shift and rotation errors. Here, we obtain results pertaining to the non-existence of approximate quantum error correcting $g$-gapped single-mode bosonic codes with respect to Gaussian dephasing errors. We show that by using arbitrarily many modes, $g$-gapped multi-mode codes can yield good approximate quantum error correction codes for any finite magnitude of Gaussian dephasing and amplitude damping errors. 
\end{abstract}

 \section{Introduction}

Traditionally, quantum error correction is studied on physical systems comprising of multiple particles, where each particle has a finite number of states and admits an interpretation as a qudit.
The subspace within multiple qudits which encodes quantum information to gain resilience against noise is known as a qudit code.
A resource required for qudit codes is the number of their underlying qudits.
This is because by increasing the number of qudits, one can correspondingly increase the number of correctible errors.
 
Since every qudit has an associated cost, it is advantageous to minimize the number of particles used by a quantum code.
However, qudit codes often use many particles.
One way to circumvent the high cost of using too many qudits is to 
encode quantum information in a few bosonic modes \cite{CLY97,WaB07,BvL16,ouyang2019permutation},
where each bosonic mode has a large number of equally-spaced energy levels.
A state on a bosonic mode lies in the span of the Fock basis
 $\{|n\>:n \in \mathbb N\}$, 
 where $\<n|m\> = \delta_{i,j}$ and $n$ in $|n\>$ counts the number of excitations.
Hence, the larger $n$ is, the more energy lies in $|n\>$.
A large number of states available in each bosonic mode allow the construction of bosonic quantum codes with resilience against various types of errors. 
Even with a single bosonic mode, 
we can have corresponding bosonic quantum codes that correct a non-trivial set of errors \cite{GKP01,BinomialCodes2016,GCB20-PhysRevX.10.011058}.
This feature combined with the possibility of preparing single-mode bosonic codes \cite{hu2019quantum} have contributed to renewed interest in bosonic quantum codes \cite{BinomialCodes2016,GCB20-PhysRevX.10.011058,terhal2020towards,hanggli2020enhanced}.

On a single bosonic mode, we can always express logical codewords in the form $\ket{j_L}:=\sum_{k \ge 0 } c_{j,k} |n+g k\>$, where the integer $g$ relates to the number of correctible number-shift errors and $n$ is any constant shift. This is fully general because the trivial $g=1$ case allows for arbitrary logical codewords.  However, many families of interest such as rotation-invariant bosonic codes~\cite{GCB20-PhysRevX.10.011058} and binomial codes~\cite{BinomialCodes2016},  have encoded states with support on Fock states that are gapped by a constant integer $g>1$ apart.  We call such codes, $g$-gapped codes. On $N$ modes, such states have the form
\begin{align}
 \ket{j_L} := \sum_{{\bf k} \in \mathbb N^N}
     c_{j,{\bf k}} 
     |{\bf n}+ g {\bf k}\>,\label{def:g-gapped-codewords}
\end{align}
where ${\bf n} \in \mathbb N^N$ is any constant shift on multiple modes.
By carefully choosing the coefficients for the logical codewords, these $g$-gapped codes can also correct phase errors. 

By correcting number-shift and phase-shift errors, rotation-invariant codes are analogous to GKP codes \cite{GKP01} that correct small displacement errors in the position and momentum quadratures. For single-mode GKP code, there is a trade-off between these quadratures.  In principle, GKP codes can tolerate an arbitrary amount of position displacement noise, though at the cost of an increased vulnerability to momentum displacement noise. We are not aware of any no-go theorems enforcing such a trade-off for single-mode bosonic codes, though it is widely believed that there is a general principle at work here.  Similarly, Grimsmo \textit{et al.} \cite{GCB20-PhysRevX.10.011058} have argued that there is a trade-off between number and phase shift, but without any strict no-go statements.

Given any bosonic quantum codes on a single-mode, one can evaluate its performance with respect to a noisy quantum channel $\mathcal N$. 
Since logical error rates cannot be perfectly suppressed for realistic error channels $\mathcal N$, such codes are invariably approximate quantum error correcting (AQEC) codes with respect to $\mathcal N$.
In what follows we define AQEC codes \cite{LNCY97,BaK02,Kle07,Fletcher08,CeO10,Tys10,ouyang2014permutation,kubica2020using,zhou2020new}.
\begin{definition}[AQEC criterion] \label{AQEC}
Given a non-negative number $\epsilon$ and a noise channel $\mathcal N$, 
we say that a quantum code $\mathcal C$ is $(\epsilon,\mathcal N)$-AQEC if and only if there exists a quantum channel $\mathcal R$ such that for every $|\psi\> \in \mathcal C$, we have
\begin{align}
\frac 1 2 \| (\mathcal R ( \mathcal N(|\psi\>\<\psi| )) - |\psi\>\<\psi|  \|_1 \le \epsilon.\label{AQEC:first}
\end{align}

\end{definition}

Numerous works have focused on the existence of AQEC codes  \cite{LNCY97,BaK02,Kle07,Fletcher08,CeO10,Tys10,ouyang2014permutation}
and non-existence of covariant AQEC codes \cite{kubica2020using,zhou2020new}, but less is known about non-existence of AQEC codes in general. The first general non-existence of AQEC qubit codes was recently addressed in the context of amplitude damping errors. Roughly speaking, using the language of quantum weight enumerators and linear programming bounds, when $\epsilon$ is too small, amplitude damping qubit codes do not exist \cite{ouyang2020linear}.  In contrast, we lack results for the non-existence of AQEC codes on any finite number of bosonic modes.

Intuitively, we expect that when the noise described by $\mathcal N$ becomes too severe, there cannot exist $(\epsilon,\mathcal N)$-AQEC bosonic codes when $\epsilon$ is sufficiently small.  An example of a noise channel that severely decoheres a bosonic mode is one that introduces random displacements with large variances in both the position and momentum quadratures.
Indeed, any single-mode bosonic code with too little energy is doomed to be useless under the effect of such a noise channel, because such large random displacement errors will effectively apply a one-time pad on the code \cite{Ouyang_2020_PRR}. 
However, number and phase shift errors with respect to AQEC codes remain to be studied.

The trade-off between resilience to number-shift and rotation errors, while recognized to be an important problem, is not well-understood \cite{GCB20-PhysRevX.10.011058}. 
In this paper, we address this by studying the performance of $g$-gapped codes with respect to a Gaussian dephasing channel $\mathcal E_\sigma$ with standard deviation $\sigma$ with the following action
 \begin{align} \label{DephaseModel}
\mathcal E_\sigma ( |n\>\<m|) &= \exp(-(m-n)^2 \sigma^2/2) |n\>\<m|,
\end{align}
as was also considered in Ref.~\cite{GCB20-PhysRevX.10.011058}.
The noise channel $\mathcal E_\sigma$ could arise as the result of a time $t$ evolving according to the master equation
\begin{equation}
    \dot{\rho} = \kappa \left( \hat{n} \rho  \hat{n} - \frac{1}{2}  \hat{n}^2 \rho - \frac{1}{2} \rho \hat{n}^2     \right)
\end{equation}
where $\sigma = \kappa t /4$ and $\hat n = \sum_{n \ge 0} n |n\>\<n|$ denotes the number operator. 
This dissipative process could be the result of rapid oscillations in the frequency of the bosonic system. For optical cavities, such rapid oscillations are regarded as the second most common noise process after photon loss~\cite{albert2018performance}.  The master equation could also arise from a weak interaction of the form $\hat{n}\otimes B$ (e.g. a cross-Kerr nonlinearity) where $B$ is some operator acting on an environment that is evolving on a much faster time scale, so that the Markov approximation holds. This dissipative process also appears as the $\kappa_\phi$ process in Ref.~\cite{chamberland2020building}.


Alternatively, our dephasing process could be the result of imperfect calibration of the system Hamiltonian, where if we average over a distribution of calibration errors we have a channel
\begin{align}
\mathcal E_\sigma(\rho)
=
\int_{-\infty}^\infty
p(\theta) e^{-i\theta \hat n} \rho e^{i\theta \hat n}
d\theta
\label{def:Gaussian-dephasing},
\end{align}
with a Gaussian probability distribution $p(\theta) =e^{-\theta^2 / (2\sigma^2)} /( \sigma \sqrt {2 \pi} )$.

We give our main result on the non-existence of AQEC $g$-gapped bosonic codes in the following theorem. This result applies not just to $g$-gapped bosonic codes on a single bosonic mode, but also to those on multiple bosonic modes.
\begin{theorem} \label{MainThm-v2} 
Let $\mathcal C_g$ be any \smash{$N$-mode} bosonic quantum code with codewords of the form \eqref{def:g-gapped-codewords} where $g$ is a positive integer. 
Let $\mathcal{E}_{\sigma}$ be a Gaussian dephasing channel with variance $\sigma$ as defined in Eq.~\eqref{DephaseModel}, and let $\mathcal N = \mathcal E_\sigma^{\otimes N}$.  
Then, the code $\mathcal C_g$ is not 
$(\epsilon, \mathcal N)$-AQEC (recall Def.~\ref{AQEC}) for any 
$\epsilon < \epsilon_{g,\sigma}$ where
\begin{equation} 
\label{EpsilonGsigma}
    \epsilon_{g,\sigma} := 1- \frac{1}{\sqrt{2}} \left(
    1+2 
    \sum_{
        \substack{
            {\bf k}\in \mathbb N ^N\\
            {\bf k} \neq 0\\
        }
    } 
    e^{- g^2 \|{\bf k}\|_2^2 \sigma^2 / 2}
    \right) .
\end{equation}
Furthermore,  if
\begin{equation} \label{tradeoff}
        g \sigma \ge 
\sqrt{-2 \log\left(  1-2^{3N/2}(2+\sqrt 2)^{-1/N} \right) },
\end{equation}
we have the non-trivial bound $\epsilon_{g,\sigma}>0$.
\end{theorem}
The trade-off in Eq.~\eqref{tradeoff} highlights an analogous phenomenon to the Heisenberg uncertainty relation in a quantum error correction setting.  We plot our bound on $\epsilon_{g,\sigma}$ in Fig.~\ref{fig:plot}. 
Note that $\lim_{\sigma}    \epsilon_{g,\sigma} \to 1-1/\sqrt 2 \approx 0.2929$, which explains the horizontal asymptotes in Fig.~\ref{fig:plot}.
In a subsequent plot given by Fig.~\ref{fig:binomial-code-comparison}, we compare the performance of $g$-gapped binomial codes \cite{BinomialCodes2016} with our no-go bounds when $g=64$.

\begin{figure}
  \centering
  \includegraphics[width=0.7\linewidth]{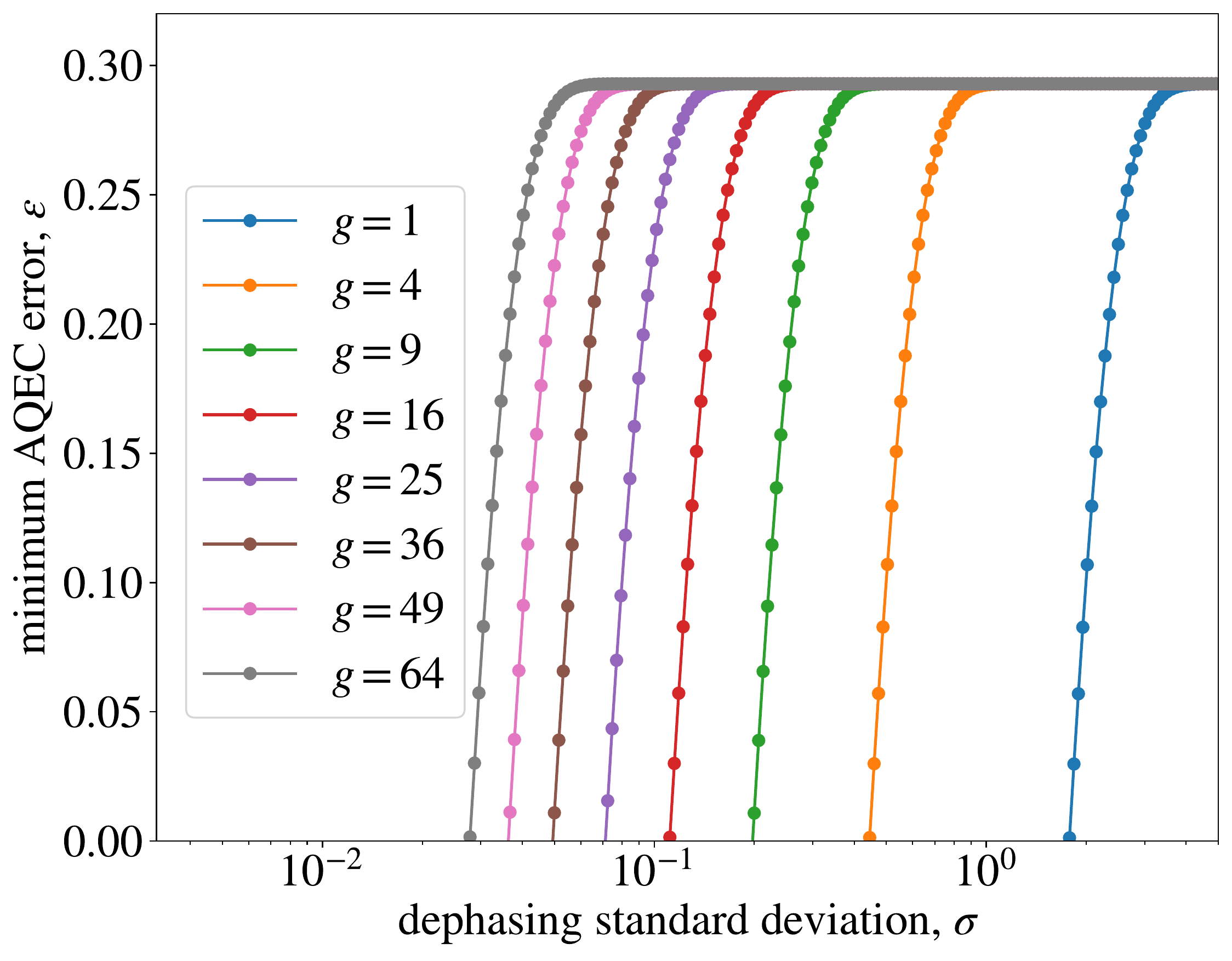}
  \caption{For single-mode $g$-gapped bosonic codes, we plot the value of $\epsilon_{g,\sigma}$ defined in Eq.~\eqref{EpsilonGsigma} as a function of the noise parameter $\sigma$ and the number shift distance $g$.  By virtue of Thm.~\ref{MainThm-v2} this gives upper bounds on $\epsilon$ for which $\epsilon$-AQEC bosonic quantum codes on a single mode exist against the noise channel $\mathcal E_\sigma$.}
  \label{fig:plot}
\end{figure}

Our no-go results for AQEC on bosonic codes with a finite number of modes are complementary to those relating to the achievability of AQEC bosonic codes.
Our results extend the converse bounds on quantum capacities for finite block-length qudit codes \cite{MaW14-converse,tomamichel2016quantum} to the bosonic setting, 
in the special case where we consider the effect of Gaussian dephasing errors on $g$-gapped bosonic codes.
In this context, our results give converse bounds on $g$-gapped bosonic codes with finite block-length.
To the best of our knowledge, no-go results for AQEC bosonic codes have never before been addressed, and our results can be interpreted to provide the first converse bounds for bosonic codes in an AQEC setting.
The simplicity of our result's proof as compared to related results on qubit-codes \cite{MaW14-converse,tomamichel2016quantum,ouyang2020linear} leads us to believe that our methods can pave the way ahead to provide more accessible results for converse bounds for bosonic codes. 

Besides no-go results for AQEC on bosonic codes, we prove that $g$-gapped codes suffice to correct purely number-shift errors. We also provide bounds on what $g$-gapped codes can achieve asymptotically under the influence of a convex combination of dephasing errors and amplitude damping errors. Namely, we show that AQEC bosonic codes on arbitrarily many modes can have with vanishing failure probabilities for quantum channels that introduce convex combinations of dephasing and amplitude damping errors on each bosonic mode. This reveals an additional trade-off between the number of modes used and the threshold with respect to $\sigma$.
 
\section{No-go for AQEC $g$-gapped bosonic codes}
In this section, we prove Thm.~\ref{MainThm-v2} which applies not just for single-mode $g$-gapped bosonic codes, but also to multi-mode $g$-gapped bosonic codes. Now, Def.~\ref{AQEC} can be equivalently stated in the following converse form. 
\begin{definition}[Alternative AQEC criterion] \label{AQEC2}
Given a non-negative number $\epsilon$ and a noise channel $\mathcal N$, 
we say that a quantum code $\mathcal C$ is not $(\epsilon,\mathcal N)$-AQEC if and only if for every $\mathcal R$ there exists a $|\psi\> \in \mathcal C$, so that
\begin{align}
 \frac{1}{2} \| (\mathcal R ( \mathcal N(|\psi\>\<\psi| )) - |\psi\>\<\psi|  \|_1 
 > \epsilon
 . \label{AQEC:second}
\end{align}
\end{definition}
This form of the AQEC criterion is more natural when proving an impossibility result. We start by establishing a simple lemma that relates the AQEC property to trace-norm closeness for pairs of orthogonal noisy codewords. 
\begin{lemma}\label{lem:connection}
Let $\delta$ be a positive real and $\mathcal N$ be a noise channel, and $\mathcal C$ be a quantum code. 
Suppose that there exist orthogonal density matrices $\rho_1$ and $\rho_2$ supported on the codespace $\mathcal C$ such that $\smash{\|\mathcal N(\rho_1 - \rho_2) \|_1 \le \delta}$.
Then $\mathcal C$ is not $(\epsilon,\mathcal N)$-AQEC where $\epsilon = 1- \delta/2$.
\end{lemma}
\begin{proof}
Since quantum channels are contractive with respect to the trace-norm, 
\begin{align}
\| \mathcal R (\mathcal N ( \rho_1 - \rho_2 ) ) \|_1
\le \| \mathcal N ( \rho_1 - \rho_2 )  \|_1 \le \delta.
\label{eq:lem-eq-1}
\end{align}
By the triangle inequality, we can see that
\begin{align}
&\| (\rho_1 - \rho_2)\|_1 -
\| \mathcal R ( \mathcal N( \rho_1 - \rho_2) ) \|_1
\notag\\
&\le  \| (\rho_1 - \rho_2)  -  \mathcal R ( \mathcal N( \rho_1 - \rho_2) ) \|_1.
\label{eq:lem-eq-2}
\end{align}
Since $\rho_1$ and $\rho_2$ are orthogonal, we have $\|\rho_1 - \rho_2\|_1 = 2$. 
Then using \eqref{eq:lem-eq-2} with \eqref{eq:lem-eq-1}, 
we get 
\begin{align}
\| (\rho_1 - \rho_2)  -  \mathcal R ( \mathcal N( \rho_1 - \rho_2) ) \|_1
\ge &2-\delta.
\end{align}
By linearity of quantum channels, we have
\begin{align}
& (\rho_1 - \rho_2)  -  \mathcal R ( \mathcal N( \rho_1 - \rho_2) )\notag\\
 = &
 (\rho_1 -\mathcal R ( \mathcal N( \rho_1 ))) - (\rho_2  - \mathcal R ( \mathcal N( \rho_2))). 
\end{align}
Using the triangle inequality, it follows that
\begin{align}
&\|  (\rho_1 -\mathcal R ( \mathcal N( \rho_1 ))) - (\rho_2  - \mathcal R ( \mathcal N( \rho_2)))\|_1
\notag\\
\le& 
\|   \rho_1 -\mathcal R ( \mathcal N( \rho_1 ) )  \|_1
+
\|    \rho_2  - \mathcal R ( \mathcal N( \rho_2))\|_1.
\end{align}
Hence
\begin{align}
\|   \rho_1 -\mathcal R ( \mathcal N( \rho_1 ) )  \|_1
+
\|    \rho_2  - \mathcal R ( \mathcal N( \rho_2))\|_1
\ge& 2 -\delta.
\end{align}
Hence, either 
\begin{align}
\|   \rho_1 -\mathcal R ( \mathcal N( \rho_1 ) )  \|_1 
\ge& 1 -\delta/2,
\end{align}
or
\begin{align}
\|    \rho_2  - \mathcal R ( \mathcal N( \rho_2))\|_1
\ge& 1 -\delta/2,
\end{align}
from which the result follows.
\end{proof}

Our next lemma shows that dephasing noise can lead to trace-norm closeness for pairs of orthogonal noisy codewords.
 
\begin{lemma} \label{lem:main-result-v2} 
For any positive integer $g$, let $\mathcal C_g$ be any \smash{$N$-mode} bosonic quantum code with codewords of the form \eqref{def:g-gapped-codewords}.
Let $\mathcal N = \mathcal E_\sigma^{\otimes N}$.
Then, there always exists orthogonal pure states $|\phi\>$ and $|\psi\>$ in $\mathcal C_g$ such that 
\begin{align}
\|\mathcal N(  |\phi\>\<\phi| - |\psi\>\<\psi| ) \|_1  \le   
 \sqrt 2 + 2\sqrt 2 \sum_{\substack{ {\bf k}\in \mathbb N^N \\ {\bf k} \neq 0 \\}} e^{-g^2 \|{\bf k}\|_2^2 \sigma^2/2}.
\notag
\end{align}
\end{lemma}
\begin{proof}[Proof of Lemma \ref{lem:main-result-v2}] 
Pick any pair of orthogonal codewords that we will call $|0_L\>$ and $|1_L\>$ and expand in the dephasing-basis, which for us is the
number basis, so that 
$|0_L\> = \sum_{{\bf k} \in \mathbb N^N} a_{\bf k} |{\bf k}\>$ and 
$|1_L\> = \sum_{{\bf k} \in \mathbb N^N} b_{\bf k} |{\bf k}\>$, 
where normalisation demands that 
$\sum_{{\bf k} \in \mathbb N^N } |a_{\bf k}|^2 =
\sum_{{\bf k} \in \mathbb N^N } |b_{\bf k}|^2= 1$.
Then the following states are normalised codewords
\begin{align}
|+_L\> &= \frac{|0_L\> + |1_L\>}{\sqrt 2} = \sum_{{\bf k} \in \mathbb N^N } \frac{a_{\bf k}+b_{\bf k}}{\sqrt 2} |{\bf k}\>\\
|-_L\> &= \frac{|0_L\> - |1_L\>}{\sqrt 2} = \sum_{{\bf k} \in \mathbb N^N } \frac{a_{\bf k}-b_{\bf k}}{\sqrt 2} |{\bf k}\>\\
|+i_L\> &= \frac{|0_L\> + i |1_L\>}{\sqrt 2} = \sum_{{\bf k} \in \mathbb N^N } \frac{a_{\bf k}+ib_{\bf k}}{\sqrt 2}|{\bf k}\>\\
|-i_L\> &= \frac{|0_L\> - i |1_L\>}{\sqrt 2} = \sum_{{\bf k} \in \mathbb N^N } \frac{a_{\bf k}-ib_{\bf k}}{\sqrt 2} |{\bf k}\>.
\end{align}
Let $\rho_+ =  |+_L\>\<+_L|$, 
 $\rho_- =  |-_L\>\<-_L|$, 
 $\rho_{+i} =  |+i_L\>\<+i_L|$, 
 $\rho_{-i} =  |-i_L\>\<-i_L|$.
For the dephasing noise model of interest (recall Eq.~\eqref{DephaseModel}) and the $g$-gapped property of the bosonic code, it follows that
\begin{align}
&\mathcal N( \rho_+ - \rho_- ) \notag\\
=& 2\sum_{{\bf j},{\bf k} \in \mathbb N^N} 
\Real (a_{\bf j}^* b_{\bf k}) |{\bf j}\>\<{\bf k}| \exp(-\|{\bf j}-{\bf k}\|_2^2 \sigma^2/2)\notag \\
=&
2\sum_{{\bf j}\in \mathbb N^N} \Real (a_{\bf j}^* b_{\bf j}) |{\bf j}\>\<{\bf j}| \notag \\
&+
2\sum_{{\bf j}\in \mathbb N^N} \sum_{{\bf k} \in \mathbb N^N, {\bf k} \neq 0} \Real (a_{\bf j}^* b_{{\bf j}+g{\bf k}}) |{\bf j}\>\<{\bf j}+g{\bf k}| e^{-g^2\|{\bf k}\|_2^2 \sigma^2/2}\notag \\
&+
2\sum_{{\bf j} \in \mathbb N^N}  \sum_{{\bf k} \in \mathbb N^N,{\bf k}\neq 0}
\Real (a_{{\bf j}+g{\bf k}}^* b_{\bf j}) 
|{\bf j}+g{\bf k}\>\<{\bf j}| e^{-g^2\|{\bf k}\|_2^2 \sigma^2/2}.\notag
\end{align}
Since the trace norm of every matrix basis $|{\bf j}\>\<{\bf k}|$ is at most 1, we can apply the triangle inequality for the trace norm to get that
$\|\mathcal E_\sigma( \rho_+ - \rho_- )\|_1 $ is at most
\begin{align} 
2\sum_{{\bf j}\in \mathbb N^N}
\bigl(
 |\Real (a_{\bf j}^* b_{\bf j})|  
+
2 
\sum_{\substack{{\bf k} \in \mathbb N^N\\ {\bf k} \neq 0}}
|\Real (a_{\bf j}^* b_{{\bf j}+g{\bf k}})| 
e^{-g^2\|{\bf k}\|_2^2 \sigma^2/2}
\bigr).\notag
\end{align}
Similarly, $\|\mathcal E_\sigma( \rho_{+i} - \rho_{-i} ) \|_1 $ is at most
\begin{align} 
2\sum_{{\bf j}\in \mathbb N^N} 
\bigl(
|\Imag (a_{\bf j}^* b_{\bf j})|  
+
2
\sum_{\substack{{\bf k} \in \mathbb N^N\\ {\bf k} \neq 0}}
|\Imag (a_{\bf j}^* b_{{\bf j}+g{\bf k}})| 
e^{-g^2\|{\bf k}\|_2^2 \sigma^2/2}
\bigr).\notag
\end{align}
Now, for any complex number $z$ we have that 
$|\Real(z)| + |\Imag(z)| \le \sqrt 2|z|$. Therefore
\begin{align}
&\|\mathcal E_\sigma( \rho_{+} - \rho_{-} ) \|_1 
+ \|\mathcal E_\sigma( \rho_{+i} - \rho_{-i} ) \|_1\notag\\
& \le 
2\sqrt 2\sum_{{\bf j}\in \mathbb N^N} |a_{\bf j}^* b_{\bf j}| 
+
4\sqrt 2
\sum_{{\bf j}\in \mathbb N^N}
\sum_{\substack{{\bf k} \in \mathbb N^N\\ {\bf k} \neq 0}}
|a_{\bf j}^* b_{{\bf j}+ g {\bf k}}|  e^{-g^2 \|{\bf k}\|_2^2 \sigma^2/2}.
\notag
\end{align}
Now we like to apply the Cauchy-Schwarz inequality over the above summation indices ${\bf j}$.
Let 
$|\psi_0\> = \sum_{{\bf j}\in \mathbb N^N} |a_{\bf j}| |{\bf j}\>$ and
$|\psi_1\> = \sum_{{\bf j}\in \mathbb N^N} |b_{\bf j}| |{\bf j}\>$.
Clearly
$\<\psi_0|\psi_0\> = 
\sum_{{\bf j}\in \mathbb N^N} |a_{\bf j}|^2 = 1$ 
and 
$\<\psi_1|\psi_1\> = 
\sum_{{\bf j}\in \mathbb N^N} |b_{\bf j}|^2 = 1$.
Hence, the Cauchy-Schwarz inequality yields
\begin{align}
\sum_{{\bf j}\in \mathbb N^N}
|a_{\bf j}^* b_{\bf j} |
=|\<\psi_0|\psi_1\>|\le 
\sqrt{\<\psi_0|\psi_0\>\<\psi_1|\psi_1\>}\le 
1.\notag
\end{align} 
Now let $|\phi_{g {\bf k}}\> = 
 \sum_{{\bf j}\in \mathbb N^N} |b_{\bf j + {\bf k}}| |{\bf j}\>$.
 Clearly, $\<\phi_{g {\bf k}} |\phi_{g {\bf k}}\> \le \<\psi_1|\psi_1\>= 1$.
Hence 
\begin{align}
\sum_{{\bf j}\in \mathbb N^N}
|a_{\bf j}^* b_{{\bf j}+g{\bf k}} | 
=& |\<\psi_0 |\phi_{g {\bf k}}\>|
\le 
\sqrt{\<\psi_0|\psi_0\>\<\phi_{g {\bf k}}|\phi_{g {\bf k}}\>}
\le 1. \notag
\end{align}
Hence, $\|\mathcal E_\sigma( \rho_{+} - \rho_{-} ) \|_1 
+ \|\mathcal E_\sigma( \rho_{+i} - \rho_{-i} ) \|_1$ is at most
\begin{align}  2\sqrt 2 
+ 4\sqrt 2  
\sum_{\substack{{\bf k} \in \mathbb N^N\\ {\bf k} \neq 0}}
e^{-g^2 \|{\bf k}\|_2^2 \sigma^2/2}.\label{eq:normsum}
\end{align}
Given two positive real numbers $a$ and $b$ and an inequality $a + b \le c$, we know that either $a \le c/2$ or $b\le c/2$. Applying this reasoning to \eqref{eq:normsum}, we get that either
\begin{align}
 \|\mathcal E_\sigma( \rho_{+} - \rho_{-} ) \|_1  
\le&   \sqrt 2 +2\sqrt 2  \sum_{\substack{{\bf k} \in \mathbb N^N\\ {\bf k} \neq 0}}
e^{-g^2 \|{\bf k}\|_2^2\sigma^2 / 2},
\end{align}
or
\begin{align}
  \|\mathcal E_\sigma( \rho_{+i} - \rho_{-i} ) \|_1 
  \le&   \sqrt 2 +2\sqrt 2 \sum_{\substack{{\bf k} \in \mathbb N^N\\ {\bf k} \neq 0}}
e^{-g^2 \|{\bf k}\|_2^2\sigma^2/2}.
\end{align}
This proves the result.
\end{proof}

Combining Lem.~\ref{lem:connection} and Lem.~\ref{lem:main-result-v2} we immediately deduce the first statement Thm.~\ref{MainThm-v2}. To prove the second result of Thm.~\ref{MainThm-v2}, we determine the values of $\sigma$ for which Thm.~\ref{MainThm-v2} is non-trivial so that $\epsilon_{g,\sigma}>0$.  
Now, we observe the following fact.
 \begin{remark}
Using a geometric series, we find that
\begin{equation}
\label{epsilonApprox-v2}
    \epsilon_{g,\sigma} \geq 1- \frac{1}{\sqrt{2}}\left( 1 + 
    2 \left(
    \frac{1}{1-N e^{-g^2\sigma^2/2}} - 1
    \right)
    \right).
\end{equation}
When $N=1$, we have 
\begin{align}
    \epsilon_{g,\sigma} \geq 1- \frac{1}{\sqrt{2}}\left( 1 + \frac{2}{e^{g^2 \sigma^2 / 2}-1}\right).\label{epsilonApprox}
\end{align}
\end{remark}
We obtain this by noting that when $r \geq 1$ we have the relaxation $\sum_{k=0}^\infty r^{-k^2} \leq \sum_{k=0}^\infty r^{-k} = 1 / (1-1/r) $. 
Hence,
\begin{align}
    \sum_{
        \substack{
            {\bf k }\in \mathbb N^N
        }
    } e^{-g^2 \|{\bf k}\|_2^2 \sigma^2/2}
    &\le
    \left(
    \sum_{
        \substack{
            {\bf k }\in \mathbb N
        }
    } e^{-g^2 k \sigma^2/2}
    \right)^N 
    =
   (1-1/r )^{-N}.\notag
\end{align}
Setting $r=\exp(g^2 \sigma^2 / 2)$ and applying to the expression for $\epsilon_{g,\sigma}$ gives the result.

To obtain upper bounds on values of $\sigma$, we let $\sigma_{\rm thres}$ denote the value where the right side of Eq.~\eqref{epsilonApprox} equals zero.  Therefore, we solve the equation
\begin{align}
 1- \frac{1}{\sqrt{2}}\left( 
 \frac{2}{(1 - e^{-g^2 \sigma_{\rm thres}^2 / 2})^N} - 1 \right) = 0.\label{eq:threshold-equation}
\end{align}
Rearranging terms in \eqref{eq:threshold-equation}, we get
\begin{align}
 e^{-g^2 (\sigma_{\rm thres})^2/2}  
 =
1-2^{3N/2}(2+\sqrt 2)^{-1/N}
\end{align}
Solving this, we get 
\begin{align}
 g \sigma_{\rm thres} &=  
\sqrt{-2 \log\left(  1-2^{3N/2}(2+\sqrt 2)^{-1/N} \right) }.
\end{align}
For small values of $N$, we find that

\begin{align}
g \sigma_{\rm thres} &\approx
 \begin{cases}
  1.87, & N=1\\
  2.19, & N=2\\
  2.36 ,& N=3\\
  2.48 ,& N=4\\
  2.56 ,& N=5.
 \end{cases}
 \end{align}

This is reminiscent of a Heisenberg uncertainty relationship between excitation errors and phase errors.  We can tighten this bound by, instead of using a geometric series to bound $\sum_{k=1}^\infty \exp(-g^2 k^2 \sigma^2/2)$, to evaluate the sum explicitly. 
Our main result hence implies that, whenever $\sigma$ is greater than $\sigma_{\rm thres}$, we cannot completely suppress the errors induced by the Gaussian dephasing channel, no matter how much energy the bosonic code has.
 
\section{The $g$-gap is sufficient for number-shift resilience}
  We now prove the existence of $g$-gapped codes encoding a single-qubit that are resilient against number-shift errors, purely by virtue of their $g$-gapped property.
 Let $\Omega$ denote a finite set of Kraus operators that induce number shifts in the Fock basis.
 In particular, every Kraus operator in $K_p \in \Omega$ has the form 
 $K_p=\sum_{j\in \mathbb N} k_{p,j} |j+u\>\<j|$ or
  $K_p=\sum_{j\in \mathbb N} k'_{p,j} |j\>\<j+u|$, corresponding to number gain and number subtraction respectively for some positive integer $u$ such that $u < g/2$ and complex coefficients $k_{p,j}$ and $k'_{p,j}$.
From the $g$-gap criterion together with the fact that $u < g/2$, if distinct logical codewords $|j_L\>$ in a $g$-gapped code are supported on a distinct set of Fock states, 
 then it is clear that for all $K,K'\in \Omega$ and for $j \neq j'$ we have
 \begin{align}
\<j_L | K^\dagger  K' | j'_L \>=0.
\end{align}
The only requirements from the Knill-Laflamme quantum error correction criterion that do not automatically follow from the $g$-gapped criterion are the non-deformation conditions given by
\begin{align}
\<j_L | K^\dagger  K' | j_L\>= \<j'_L | K^\dagger K' | j'_L\>\label{eq:nodeform}
\end{align} 
for every $K,K' \in \Omega$ and all logical codewords $| j_L\>$ and $| j'_L\>$.
In the special case where a $g$-gapped code encodes a single qubit, techniques in Ref~\cite{movassagh2020constructing} imply the existence of codes that satisfy the non-deformation condition \eqref{eq:nodeform}. To see this, consider single-mode bosonic codes. 
Define a matrix, 
\begin{align}
A &= \sum_{K,K' \in \Omega \cup\{I\}} \sum_{k \in \mathbb N} 
(
{\rm Re}(\<gk | K^\dagger  K' |gk\>)
|K,K',1\>\<k| \notag\\
&\quad+
{\rm Im}(\<gk | K^\dagger  K' |gk\>)
|K,K',2\>\<k|
)
,\label{Amatrix}
\end{align}
where $I$ denotes the identity operator.
Since every matrix element of $A$ is either ${\rm Re}(\<gk | K^\dagger  K |gk\>)$ or ${\rm Im}(\<gk | K^\dagger  K |gk\>)$ which are both real, the matrix $A$ is a real matrix.
Now let $|\xi\>$ be any non-zero real vector such that
\begin{align}
|\xi\> = \sum_{k \in \mathbb N} x_k |k\>,
\end{align}
and 
\begin{align}
A |\xi\> = 0.
\end{align}
Such non-zero vectors $|\xi\>$ always exist because $A$ has an infinite number of columns and a finite number of rows, and hence its kernel must have a positive dimension.
Then, it follows using the techniques of \cite{movassagh2020constructing} that that the code spanned by 
\begin{align}
|0_L\> = \sum_{k \in \mathbb N}x^+_k  | gk\> / \sqrt{\|{\bf x}\|_2},\notag\\
|1_L\> = \sum_{k \in \mathbb N}x^-_k   | gk\> / \sqrt{\|{\bf x}\|_2}
\end{align}
is a $g$-gapped code that corrects the number shift errors in $\Omega$, 
where $x^+_k = \max(0,x_k)$, 
$x^-_k = \max(0,-x_k)$ and $\| {\bf x} \|_2$ denotes the 2-norm of the vector $(x_k)_{k \in \mathbb N}$.

\section{Tightness of no-go results on a single bosonic mode}
 
\begin{figure}
  \centering
  \includegraphics[width=0.7\linewidth]{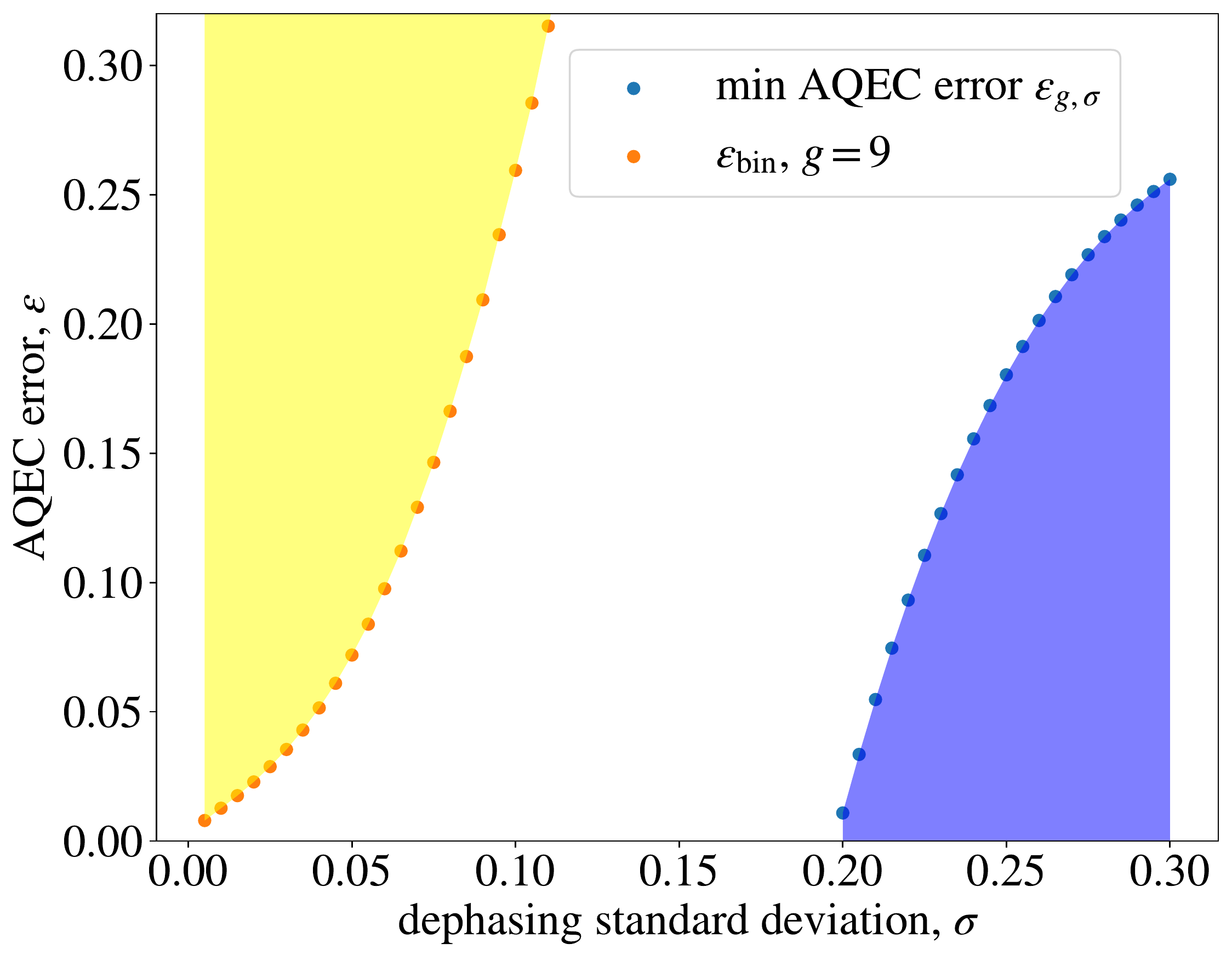}
  \caption{We compare upper bounds on $\epsilon_{\rm bin}$ using $g$-gapped binomial code with $\epsilon_{g,\sigma}$ from Thm.~\ref{MainThm-v2}, where $g=9$. Binomial codes are optimized over those that correct at least one gain and loss error, and with $D \in [2,50]$.  
For every $(\epsilon, \sigma)$ pair in the yellow region, there exists an $\epsilon$-AQEC code that corrects Gaussian dephasing with standard deviation $\sigma$.
For every $(\epsilon, \sigma)$ pair in the blue region, there does not exists any $\epsilon$-AQEC code that corrects Gaussian dephasing with standard deviation $\sigma$.
The white region is the intermediate region where we do not know if an AQEC code is possible.} 
  \label{fig:binomial-code-comparison}
\end{figure}
 
To investigate the tightness of our no-go results on $g$-gapped AQEC codes in the presence of Gaussian dephasing errors of standard deviation $\sigma$ as given in Thm.~\ref{MainThm-v2}, we investigate the performance of an explicit family of $g$-gapped bosonic codes. 
We investigate binomial codes \cite{BinomialCodes2016}, which are bosonic variants of permutation-invariant quantum codes designed for spin-systems \cite{ouyang2014permutation,OUYANG201743} can correct number-shift errors such as gain and loss errors, which are given explicitly by $(\hat a^\dagger)^j$ and $\hat a^k$. Here $\hat a$ denotes the lowering operator and $j$ and $k$ are non-negative integers that count the number of gain or loss errors.

The number of correctible phase errors in binomial codes can be understood from the series expansion of the rotation operator $e^{-i \theta \hat n}$ given by $e^{-i \theta \hat n} = {\bf 1} - i \theta \hat n  - \theta^2 \hat n^2/2! + \dots$. 
Correctibility of polynomials in $\hat n$ corresponds to the correctibility of the leading order terms in the series expansion of $e^{-i \theta \hat n}$.
Hence, the number of phase errors that a single-mode bosonic code corrects is defined to be the maximum order of the correctible polynomials in $\hat n$. 

For binomial codes encoding a single logical qubit that correct $G$ gain and $L$ loss errors, we must have $g \ge G+L+1$ \cite[Eq (7)]{BinomialCodes2016}. 
Correcting $D$ phase errors requires a maximum photon number of $n_{\rm max} = (\max\{L,G,2D\}+1)g.$
When $D\ge G/2$ and $D \ge L/2$, the maximum number of photons required simplifies to
\begin{align}
   n_{\rm max} = (2D+1)g.
\end{align}
In this scenario, we can consider a family of $g$-gapped binomial code that corrects $D$ phase errors with the logical codewords
\begin{align}
|0_L \> &= \sqrt{2^{-2D}} \sum_{\substack{0\le j \le 2D+1\\ j {\ \rm even} }}
\sqrt{\binom {2D+1}j} |gj\> \notag\\
|1_L\> &= \sqrt{2^{-2D}} \sum_{\substack{0\le j \le 2D+1\\ j {\ \rm odd} }}
\sqrt{\binom {2D+1}j} |gj\>.\label{binomial-codewords}
\end{align} 

The requirement of correcting gain and loss errors introduces the bound $D \ge 2$. Now, a Gaussian dephasing channel introduces the rotation operator $\exp(-i \theta \hat n)$ with probability $p(\theta)$ where $p(\theta)$ is the
Gaussian probability distribution with standard deviation $\sigma$ that appears in \eqref{def:Gaussian-dephasing}.
The truncation error of using $R_{\theta,D} = \sum_{j=0}^D (-i \theta \hat n)^j/j!$ in place of $\exp(-i \theta \hat n)$ on 
$|n\>$ is at most $(|\theta| n)^{D+1}/(D+1)!$.
Now let us denote $\overline R_{\theta,D} = \sum_{j=D+1}^\infty(-i \theta \hat n)^j/j!$. Clearly $e^{-i \theta \hat n } = R_{\theta,D}+\overline R_{\theta,D} $.
Any $|\psi\>$ in the binomial code's codespace can be written as $|\psi\> = a|0_L\> + b|1_L\>$, using the fact that $|a|\le 1$ and $|b|\le 1$, we find that 
\begin{align}
 \| \overline R_{\theta,D} |\psi\> \|
&\le
 \| \overline R_{\theta,D} a|0_L\> \|
+
 \| \overline R_{\theta,D} b|1_L\> \| 
 \notag\\
&\le
 \| \overline R_{\theta,D} |0_L\> \|
+
 \| \overline R_{\theta,D} |1_L\> \| 
 \notag\\
&\le
 \sqrt{2^{-2D}} \sum_{\substack{0\le j \le 2D+1\\ j {\ \rm even} }}
\sqrt{\binom{2D+1}{j}}  \| \overline R_{\theta,D} |j\> \|
\notag\\
&\quad +
 \sqrt{2^{-2D}} \sum_{\substack{0\le j \le 2D+1\\ j {\ \rm odd} }}
\sqrt{\binom{2D+1}{j}} \| \overline R_{\theta,D} |j\> \| 
 \notag\\
 &\le 
 2^{-D}  \sum_{\substack{0\le j \le 2D+1 }}
\sqrt{ \binom {2D+1} j} \frac{(|\theta|gj)^{D+1}}{(D+1)!} \notag\\
&\le\epsilon_{{\rm bin},\theta}, \label{eq:stirling}
\end{align}
where
\begin{align}
\epsilon_{{\rm bin},\theta}
&=
2^{-D}  \sum_{\substack{0\le j \le 2D+1 }}
\sqrt{ \binom {2D+1} j} \frac{ (egj/(2D+1))^{2D+1}}{\sqrt{2\pi (2D+1)}} .
\end{align}
The last inequality in \eqref{eq:stirling} follows from $n^n/n! \le \frac{e^n}{\sqrt{2\pi n}}$, which implies that  
\begin{align}
(gj)^n/n! 
&\le 
\frac{ (egj/n)^n}{\sqrt{2\pi n}}.
\end{align}
Now let 
 \begin{align}
\epsilon_{\rm bin}  =
\frac{1}{2} \int_{-\phi}^\phi p(\theta) 
\left(  \epsilon_{{\rm bin},\theta} + 2 \epsilon_{{\rm bin},\theta}^2 \right) d\theta 
 + 2 \int_{\phi}^{\infty}p(\theta) d\theta .
 \end{align}
 Then we prove the following lemma in Appendix A.
 \begin{lemma}
 Let $\mathcal C$ be the state of pure states in the span of the logical codewords of the binomial code defined in \eqref{binomial-codewords}.
There exists a recovery map $\mathcal R$ such that for 
 \begin{align}
\frac{1}{2} \max_{|\psi\> \in \mathcal C} \| \mathcal R(\mathcal N_\sigma (|\psi\>\<\psi|)) - |\psi\>\<\psi| \|_1 
&\le \epsilon_{\rm bin} 
 \end{align}
 \end{lemma}

We numerically find upper bounds to $\epsilon_{\rm bin}$ by optimizing over $2 \le D \le 50$, and $\frac{\sigma}{2}< \phi <15\sigma$ when $g=9$.
We compare these upper bounds with our lower bound on $\epsilon_{g,\sigma}$ in Figure \ref{fig:binomial-code-comparison}.  
 
\section{Achievability bounds on multimode $g$-gapped codes}
 
 \begin{figure*}
 \centering  
\subfigure[$g=1$]{
\begin{minipage}[t]{0.32\linewidth}
\centering
\includegraphics[width=1\textwidth]{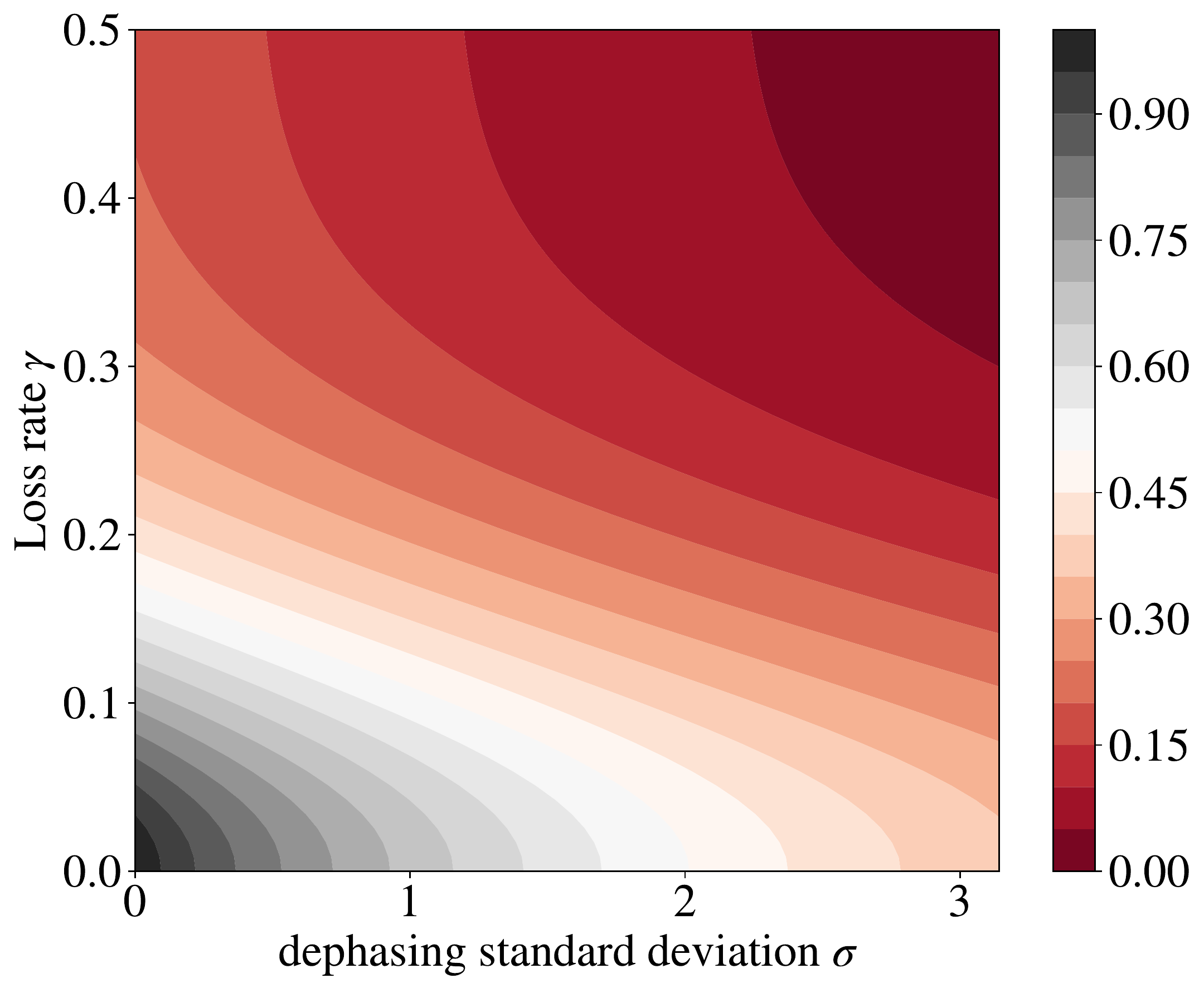} 
\end{minipage}%
}%
\subfigure[$g=10$]{
\begin{minipage}[t]{0.32\linewidth}
\centering
\includegraphics[width=1\textwidth]{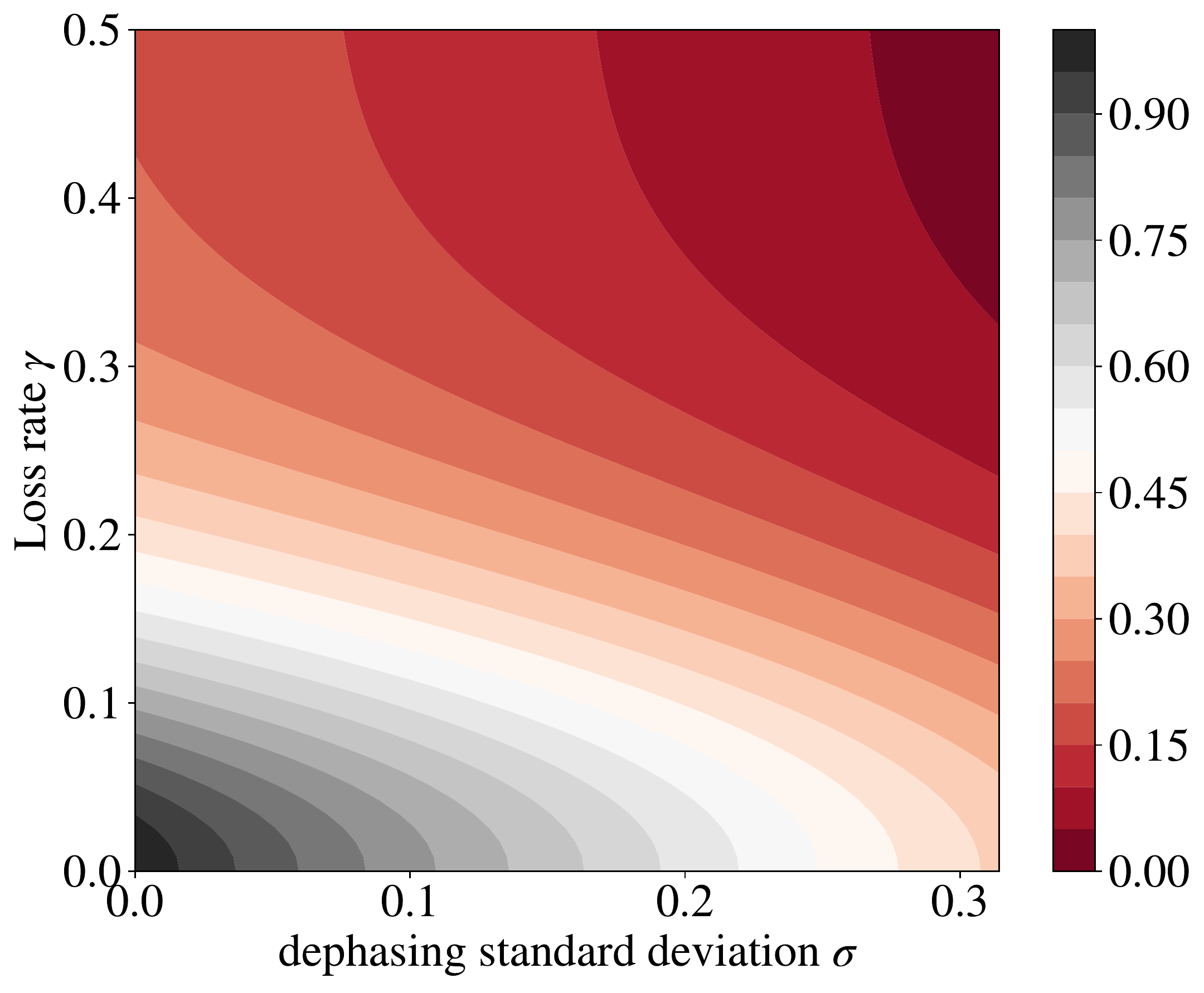} 
\end{minipage}%
}
\subfigure[$g=60$]{
\begin{minipage}[t]{0.32\linewidth}
\centering
\includegraphics[width=1\textwidth]{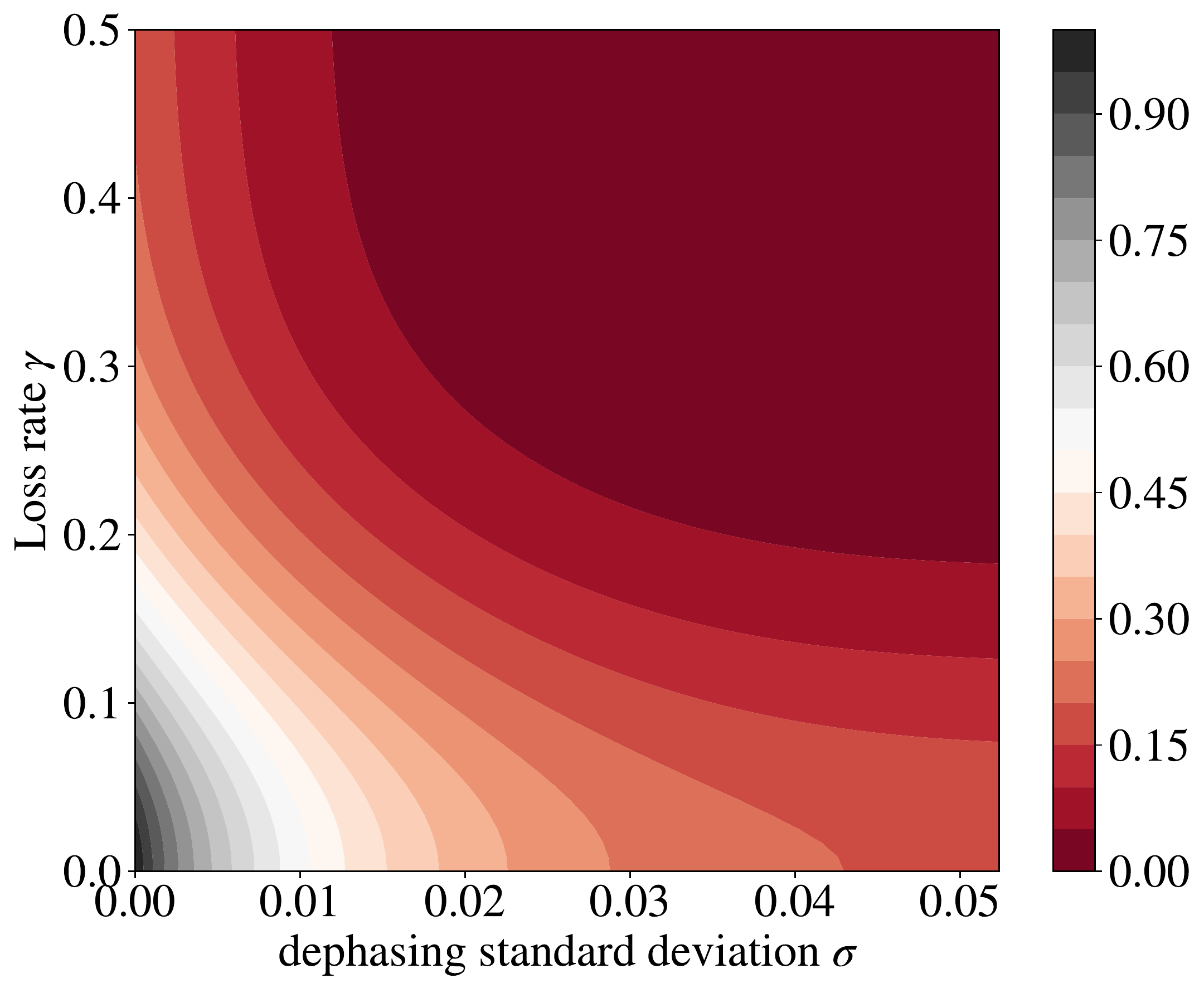} 
\end{minipage}%
}
 \caption{Lower bounds on the achievable coding rate using $g$-gapped codes when the noisy channel per bosonic mode is given by $\mathcal M = \lambda \mathcal E_\sigma + (1-\lambda) \mathcal A_\gamma$ where $\lambda=0.5$, and for different values of $g$. This shows that over a wide range of values for $\sigma$ and $\gamma$, the channel $\mathcal N$ has positive capacity with respect to $g$-gapped codes.}\label{fig:icoh}
 \end{figure*}

Evaluating the quantum capacity of an arbitrary quantum channel is difficult, and simple upper bounds for it are not necessarily tight \cite{pisarczyk2019causal}.
Evaluating the quantum capacity simplifies when the quantum channel is degradable \cite{cubitt2008structure}.  
Fortunately for us, $\mathcal E_\sigma$ is degradable, because it has simultaneously diagonal Kraus operators \cite{DeS03}.

However, consider simple multimode $g$-gapped codes that are effectively qubit codes. Let $\mathcal S$ be the span of $\{|0\>,|g\>\}$. Given $m$ modes, consider $g$-gapped codes in $\mathcal S^{\otimes m}$.
This reduces our analysis on $\mathcal E_\sigma$ to that on a simple single-qubit dephasing noise model of the form
\begin{align}
    \mathcal D( \rho ) = (1-p) \rho + p  Z \rho Z,\label{eq:effective-qubit-model}
\end{align}
where $\rho$ is supported on $\mathcal S$, $Z = |0\>\<0|  - |g\>\<g|$,
 and $\smash{1-2p = \exp(- g^2 \sigma^2 / 2)}$.
If $\sigma < \infty$, we have $p>1/2$.  
Now the channel $\mathcal{D}$ with $p>1/2$ has a non-zero quantum capacity $Q(\mathcal D)$  \cite{Rai99,DeS03,SSW08} because
\begin{align}
Q(\mathcal D) =  1+p \log_2 p + (1-p)\log_2(1-p) .
\end{align} 
For completeness, we explain how to evaluate $Q(\mathcal D)$ explicitly in Appendix B.
The quantum capacity $Q(\mathcal D)$ gives us lower bounds to the quantum capacity of $\mathcal E_\sigma$ when restricted to $g$-gapped codes.

When loss errors occur in addition to phase errors, we can model the noisy channel on each bosonic mode as 
\begin{align}
\mathcal M = \lambda \mathcal E_\sigma  + (1-\lambda) \mathcal A_\gamma,
\end{align}
where $\mathcal A_\gamma$ denote the bosonic amplitude damping channel with Kraus operators that removes $k$ excitations is given by 
\begin{align}
A_k = \sum_{m\ge k} \sqrt{\binom m k}
\sqrt{ (1-\gamma)^{m-k}\gamma	^k} |m-k\>\<m|.
\end{align} 
We show in Appendix C by applying the recovery map $\mathcal R$ with Kraus operators given by 
\begin{align}
R_j = \sum_{k=1}^\infty |kg\>\<kg-j|
\end{align}
where $j=0,1,\dots,g-1$ on the quantum channel $\mathcal M$, we obtain an effective qubit-dephasing channel on the space $\{|g\>,|2g\>\}$ with dephasing probability given by 
\begin{align}
r= \lambda p + (1-\lambda) q,
\end{align}
where
\begin{align}
q=\left(1-\sum_{k=0}^{g-1}\sqrt{\binom g k \binom {2g} k} (1-\gamma)^{3g/2-k}\gamma	^k\right)/2.
\end{align}
Hence it follows that 
\begin{align}
Q(\mathcal M) \ge  1+r \log_2 r + (1-r)\log_2(1-r) .
\end{align}
We plot in Fig.~\ref{fig:icoh} lower bounds on the achievable rates at which we can encode quantum information in $g$-gapped multi-mode codes under the channel $\mathcal M$ which models both dephasing and loss, and show that over a wide range of values for $\gamma$ and $\sigma$, this rate is positive.

Extending our analysis from qubit-dephasing channels to qudit dephasing channels can give correspondingly tighter lower bounds on the $g$-gapped quantum capacity of $\mathcal E_\sigma$.

\section{Discussions}
In summary, we have explored trade-offs on number and phase shift resilience in bosonic quantum codes. By fixing the parameter $g$ which quantifies the number-shift resilience of bosonic codes, we obtain corresponding no-go results on the correctibility of $g$-gapped errors against dephasing noise. Our results apply both to the simplest setting of a single-bosonic mode, and also multiple bosonic modes. Our work gives no-go results on what can be achieved using $g$-gapped bosonic codes both in a single-mode and a multi-mode setting.

There are several possible directions in which we believe that our work can be extended. 
Given that energy constraints in bosonic codes have received much attention in recent years \cite{holevo2004entanglement,shirokov2018energy,wilde2018energy},
it will be interesting to see how the additional introduction of energy constraints affects the trade-offs. Also, we leave the problem of tightening our no-go bounds for future work.

\section{Acknowledgements}
The authors are grateful to Robert Koenig, Lisa H\"anggli, Margret Heinze, and Barbara Terhal for fruitful discussions.
YO and EC acknowledge support from the EPSRC (Grant
No. EP/M024261/1) and the QCDA project (Grant No.
EP/R043825/1)) which has received funding from the QuantERA
ERANET Cofund in Quantum Technologies implemented
within the European Union’s Horizon 2020 Programme.  This work was completed while ETC was at the University of Sheffield.
Y.O. is supported in part by NUS startup grants (R-263-000-E32-133 and R-263-
000-E32-731), and the National Research Foundation, Prime Minister’s Office, Singapore and the Ministry of Education, Singapore under the Research Centres of Excellence programme.

\bibliography{bosonic-dephasing}{}
\bibliographystyle{ieeetr}

\appendix

\noindent{\bf A: Proof of Lemma 6}
 
 \begin{proof}
 Now for positive $\phi$, let 
\begin{align}
M_{\sigma, \phi}(\rho) 
= 
\int_{-\phi}^\phi p(\theta)
 e^{-i \theta \hat n}\rho e^{i \theta \hat n} d\theta .
\end{align}
It follows that 
\begin{align}
&\| \mathcal N_{\sigma}(|\psi\>\<\psi|)-M_{\sigma, \phi}(|\psi\>\<\psi|)  \|_1 \notag\\
=&\|(
\int_{-\infty}^{-\phi}
 p(\theta)
e^{-i \theta \hat n }
|\psi\>\<\psi| 
e^{ i \theta \hat n } d\theta \notag\\
&\quad 
+
\int_{\phi}^\infty
 p(\theta)
e^{-i \theta \hat n }
|\psi\>\<\psi| 
e^{ i \theta \hat n }  d\theta
)  \|_1 
 \notag\\
 \le &
2 \int_{\phi} ^\infty p(\theta) d\theta. \label{eq:tail-bound}
\end{align}
Now
\begin{align}
&M_{\sigma, \phi}(|\psi\>\<\psi|)  \notag\\
=&
\int_{-\phi}^\phi p(\theta)
 (R_{\theta,D}+\overline R_{\theta,D})
|\psi\>\<\psi| 
 (R_{\theta,D}^\dagger +\overline R_{\theta,D}^\dagger)
 d\theta
 \notag\\
=&
\int_{-\phi}^\phi p(\theta)
 R_{\theta,D}
|\psi\>\<\psi| 
 R_{\theta,D}^\dagger 
 d\theta
+
\int_{-\phi}^\phi p(\theta)
\overline R_{\theta,D}
|\psi\>\<\psi| 
\overline R_{\theta,D}^\dagger
 d\theta
 \notag\\
 &
+
\int_{-\phi}^\phi p(\theta)
 R_{\theta,D}
|\psi\>\<\psi| 
 \overline R_{\theta,D}^\dagger
 d\theta
 \notag\\
 &
+
\int_{-\phi}^\phi p(\theta)
 \overline R_{\theta,D}
|\psi\>\<\psi| 
 R_{\theta,D}^\dagger 
 d\theta.
 \end{align}
 Now let $\mathcal R$ be any quantum channel $\mathcal R$ that corrects $D$ phase errors.
From the trace preserving property of $\mathcal R$ and the perfect correctibility of the errors $R_{\theta,D}$, we have
 \begin{align}
 &\mathcal R \left(  R_{\theta,D} |\psi\>\<\psi|  R_{\theta,D}^\dagger \right)\notag\\
 =&
 |\psi\>\<\psi|  \<\psi| R_{\theta,D}^\dagger R_{\theta,D} |\psi \> 
 \notag\\
 =&
 |\psi\>\<\psi| \| R_{\theta,D} |\psi \>  \|^2  \notag\\
 =&
 |\psi\>\<\psi| \left(1 - \| \overline R_{\theta,D} |\psi \>  \|^2  \right).\label{eq:phase-correctibility}
 \end{align}
 Now by the triangle inequality, we have 
\begin{align}
 &\| \mathcal R(\mathcal N_\sigma (|\psi\>\<\psi|))-  |\psi\>\<\psi|  \|_1  \notag\\
\le&
\| \mathcal R( \mathcal M_{\sigma,\phi} (|\psi\>\<\psi|)) -  |\psi\>\<\psi| \|_1  \notag\\
&\quad  + 
 \| \mathcal R( (\mathcal N_\sigma  - \mathcal M_{\sigma,\phi}) (|\psi\>\<\psi|)) \|_1 .
 \end{align}
 Since $\mathcal R$ must be trace preserving and $(\mathcal N_\sigma  - \mathcal M_{\sigma,\phi})$ is completely positive, we have that 
$ \| \mathcal R( (\mathcal N_\sigma  - \mathcal M_{\sigma,\phi}) (|\psi\>\<\psi|)) \|_1  
=
 \| (\mathcal N_\sigma  - \mathcal M_{\sigma,\phi}) (|\psi\>\<\psi|) \|_1  $.
Hence using \eqref{eq:tail-bound} we find that   %
\begin{align}
 \| \mathcal R(\mathcal N_\sigma (|\psi\>\<\psi|)) - |\psi\>\<\psi| \|_1 
&\le
\| \mathcal R( \mathcal M_{\sigma,\phi} (|\psi\>\<\psi|)) -  |\psi\>\<\psi| \|_1  \notag\\
 &+ 2 \int_{\phi}^{\infty}p(\theta) d\theta .
 \end{align}
 Using \eqref{eq:phase-correctibility}, it follows that 
\begin{align}
& \| \mathcal R(\mathcal N_\sigma (|\psi\>\<\psi|)) - |\psi\>\<\psi| \|_1 \notag\\
\le&
\left\| \int_{-\phi}^\phi p(\theta)(1-\| \overline R_{\theta,D} |\psi\> \|^2) |\psi\>\<\psi| d\theta
-  |\psi\>\<\psi| \right\|_1 
\notag\\
&
+
\left\|\int_{-\phi}^\phi p(\theta) R_{\theta,D} |\psi\>\<\psi| \overline R_{\theta,D} ^\dagger d\theta  \right\|_1  \notag\\
&+
\left\|\int_{-\phi}^\phi p(\theta)\overline R_{\theta,D} |\psi\>\<\psi| R_{\theta,D} ^\dagger d\theta  \right\|_1 
 + 2 \int_{\phi}^{\infty}p(\theta) d\theta 
 \notag\\
\le&
\int_{-\phi}^\phi p(\theta) \left\| (1-\| \overline R_{\theta,D} |\psi\> \|^2) |\psi\>\<\psi| 
-  |\psi\>\<\psi| \right\|_1 d\theta
\notag\\
&
+
\int_{-\phi}^\phi p(\theta) \left\|R_{\theta,D} |\psi\>\<\psi| \overline R_{\theta,D} ^\dagger \right\|_1 d\theta  \notag\\
&+
\int_{-\phi}^\phi p(\theta)\left\|\overline R_{\theta,D} |\psi\>\<\psi| R_{\theta,D} ^\dagger \right\|_1 d\theta  
 + 4 \int_{\phi}^{\infty}p(\theta) d\theta 
 \notag\\
=&
\int_{-\phi}^\phi p(\theta) \| \overline R_{\theta,D} |\psi\> \|^2 d\theta
\notag\\
&+
\int_{-\phi}^\phi p(\theta) 
\tr \sqrt{  
\overline R_{\theta,D} |\psi\>\<\psi|  R_{\theta,D} ^\dagger
R_{\theta,D} |\psi\>\<\psi| \overline R_{\theta,D} ^\dagger  } d\theta  
\notag\\
&+
\int_{-\phi}^\phi p(\theta)
\tr \sqrt{  
 R_{\theta,D}|\psi\>\<\psi| \overline R_{\theta,D} ^\dagger 
\overline R_{\theta,D} |\psi\>\<\psi| R_{\theta,D} ^\dagger } 
d\theta   \notag\\
&
 + 4 \int_{\phi}^{\infty}p(\theta) d\theta .
 \end{align}
 Now it is clear that for any matrix $M$ we have that $\tr\sqrt{M|\psi\>\<\psi|M^\dagger} = \sqrt{\<\psi|M^\dagger M|\psi\>} = \| M | \psi\>  \|$. This is because $M|\psi\>\<\psi|M^\dagger$ is a rank one matrix, with eigenvector $M|\psi\>$ and eigenvalue $ \| M | \psi\>  \|$. Using this fact, we find that 
\begin{align}
 &\| \mathcal R(\mathcal N_\sigma (|\psi\>\<\psi|)) - |\psi\>\<\psi| \|_1 \notag\\ 
\le &
\int_{-\phi}^\phi p(\theta) \| \overline R_{\theta,D} |\psi\> \|^2 d\theta \notag\\
&+
2\int_{-\phi}^\phi p(\theta)
 \| \overline R_{\theta,D} |\psi\>\| \|  R_{\theta,D} |\psi\>\|
d\theta  \notag\\
&
 + 4 \int_{\phi}^{\infty}p(\theta) d\theta . 
 \end{align}
 Now $\| R_{\theta,D } |\psi\> \|  
 =  \|  e^{-i\theta \hat n} |\psi\> - \overline R_{\theta,D} |\psi\> \|
 \le 1 + \| \overline R_{\theta,D } |\psi\> \|  $. 
 It follows that
 \begin{align}
 &\| \mathcal R(\mathcal N_\sigma (|\psi\>\<\psi|)) - |\psi\>\<\psi| \|_1  \notag\\
\le&
\int_{-\phi}^\phi p(\theta) \left(  \| \overline R_{\theta,D} |\psi\> \|+ 2 \| \overline R_{\theta,D} |\psi\> \|^2 \right) d\theta 
 + 4 \int_{\phi}^{\infty}p(\theta) d\theta ,
 \end{align}
 and the result follows.
 \end{proof}

\noindent{\bf B The quantum capacity $Q(\mathcal D)$}

\noindent We remark that the quantum capacity of generalized dephasing channels is known to be its Rains information \cite[Proposition 10]{tomamichel2016strong}. 

Kraus operators of the dephasing channel $\mathcal D$ can be written as 
\begin{align}
    D_0 &= \sqrt{1-p}  |0\>\<0| + \sqrt{1-p}|g\>\<g|\\
    D_1 &= \sqrt{p}  |0\>\<0| - \sqrt{p}|g\>\<g|,
\end{align}
and the Kraus operators of the complementary channel $\hat{\mathcal D}$ are
\begin{align}
    R_0 &= \sqrt{1-p}  |0\>\<0|+ \sqrt{p}  |0\>\<0|\\
    R_1 &= \sqrt{1-p}|g\>\<g| - \sqrt{p}|g\>\<g|.
\end{align}
Now let us denote a diagonal state $\tau_r$ to be given by
\begin{align}
    \tau_r = (1- r) |0\>\<0| + r |g\>\<g| .
\end{align}
Denoting $S(\rho) = - \tr( \rho \log \rho)$ as the von-Neumann entropy, 
the coherent information of $\mathcal D$ is given by 
\begin{align}
&I_{\rm coh} (\mathcal D ,  \tau_r) 
=
    S( \mathcal D (\tau_r))  - 
    S(\hat { \mathcal D } (\tau_r )).
\end{align}
Because the Kraus operators of $\mathcal D$ are diagonal, the coherent information can be maximized using only diagonal input states \cite{cubitt2008structure}. 
The symmetry that allows this is the covariance of the dephasing channel with respect to diagonal operators, coupled with the concavity of the coherent information with respect to input states because of its degradability \cite[Corollary 4.3]{ouyang2014channel}. 
We will show that this coherent information is maximized by the maximally mixed state. 

Since $\mathcal D$ is degradable, we know that it is a concave function with respect to its argument which is a density matrix. 
Therefore,
$I_{\rm coh} (\mathcal D )  = \max_\rho I_{\rm coh} (\mathcal D,  \rho) $ is optimized whenever 
\begin{align}
    \frac{ d}{dr } I_{\rm coh} (\mathcal S , \tau_r  ) = 0
    \label{eq:zero-gradient}
\end{align}
and $r$ is contained in the open interval $(0,1)$.
Indeed, we can verify that \eqref{eq:zero-gradient} holds for $r=1/2$. Hence the maximally mixed state on $\mathcal S$ maximizes the coherent information.

Next we find that for all $0\le p\le 1/2$, we have
\begin{align}
&I_{\rm coh} (\mathcal D ,  \tau_{1/2})
=
    \log(2-2p)-2p \tanh^{-1}(1-2p),
    \label{eq:capacity-lower-bound}
\end{align}
where
\begin{align}
    \tanh^{-1}(z) = \frac{1}{2}\left(
    \log(1+z) - \log(1-z)
    \right).
\end{align}
Simplifying \eqref{eq:capacity-lower-bound},
we get 
\begin{align}
I_{\rm coh} (\mathcal D ,  \tau_{1/2})
= \left(
1 + p \log p + (1-p)\log(1-p)
\right)\log 2,
\end{align}
which is in agreement with \cite{Rai99,SSW08} when expressed in nats.
\newline
  
\noindent{\bf C: Reduction to an effective qubit-dephasing channel}
\newline
\noindent Here we reduce the channel $\mathcal M$ to an effective qubit-dephasing channel by constraining the bosonic states to be supported on only $|g\>$ and $|2g\>$ on each bosonic mode. 

First note that the recovery operation $\mathcal R $ is indeed a quantum channel because
\begin{align}
R_j^\dagger R_j 
&= 
\sum_{j=0}^{g-1} \sum_{k,k'=1}^\infty |kg-j\> \<kg| k'g\>\<k'g-j|\notag\\
&=
\sum_{j=0}^{g-1} \sum_{k=1}^\infty |kg-j\> \<kg-j|
={\bf 1}.
\end{align}
Second, we compose $\mathcal R$ with $\mathcal M$ to obtain $\mathcal R \circ \mathcal M$. Clearly $\mathcal R \circ \mathcal D = \mathcal D$. Hence it remains to evaluate $\mathcal R \circ \mathcal A_\gamma$ on our allowed codespace. 
Consider the projector
\begin{align}
\Pi = \sum_{a \ge 0}|ag\>\<ag| .
\end{align}
Then we can see that  
\begin{align}
R_j A_k \Pi
&= 
\delta_{j,k}
\sum_{a=1}^\infty |ag\>\<ag| \sqrt{\binom {ag} k} \sqrt{ (1-\gamma)^{ag-k}\gamma	^k} .
\end{align}
This implies that the effective amplitude damping Kraus operators on the space spanned by 
$\{|g\>,|2g\>\}$ are given by
\begin{align}
\bar R_k =& 
|g\>\<g|  \sqrt{\binom { g} k} \sqrt{ (1-\gamma)^{ g-k}\gamma	^k} \notag\\
&
+
|2g\>\<2g|  \sqrt{\binom {2g} k} \sqrt{ (1-\gamma)^{2g-k}\gamma	^k} .
\end{align}
Third, to see how $\bar R_0 ,\dots, \bar R_{g-1}$ model an effective dephasing channel, 
note that 
\begin{align}
\sum_{k=0}^{g-1} \bar R_k  (|g\>\<g|)   \bar R_k^\dagger
&=|g\>\<g|\\
\sum_{k=0}^{g-1} \bar R_k  (|2g\>\<2g|)   \bar R_k^\dagger
&=|g\>\<g|\\
\sum_{k=0}^{g-1} \bar R_k  (|g\>\<2g|)   \bar R_k^\dagger
&=|g\>\<2g| \xi
\\
\sum_{k=0}^{g-1} \bar R_k  (|2g\>\<g|)   \bar R_k^\dagger 
&=|2g\>\<g| \xi
,
\end{align}
where 
\begin{align}
\xi = \sum_{k=0}^{g-1} \sqrt{\binom g k \binom {2g} k} (1-\gamma)^{3g/2-k}\gamma	^k,
\end{align}
and hence the channel modeled by $\overline{ \mathcal R}$ with Kraus operators $\bar R_0,\dots, \bar R_{g-1}$ is an effective dephasing channel with dephasing probability $q$ where
\begin{align}
1-2q = \sum_{k=0}^{g-1}\sqrt{\binom g k \binom {2g} k} (1-\gamma)^{3g/2-k}\gamma	^k.
\end{align}
Therefore we can see that 
\begin{align}
\lambda \mathcal D + (1-\lambda) \overline{\mathcal R}
\end{align}
has dephasing probability $r$ where
\begin{align}
1-2r = \lambda(1-2p)+(1-\lambda)(1-2q).
\end{align}

\end{document}